\pdfoutput=1
\documentclass{birkjour}
\usepackage[noadjust]{cite}
\usepackage{xcolor}
\newcommand{\ot}{\otimes}
\newcommand{\tp}[1]{^{\otimes #1}}    % tensorial power

\DeclareMathOperator{\id}{id}
\DeclareMathOperator{\supp}{supp}

\newcommand{\Cl}{\mathbb{C}}
\newcommand{\Rl}{\mathbb{R}}
\newcommand{\Nl}{\mathbb{N}}
\newcommand{\Zl}{\mathbb{Z}}

\newcommand{\B}{\mathcal{B}}

\newcommand{\M}{\mathcal{M}}

\newcommand{\Hil}{\mathcal{H}}

\newcommand{\Om}{\Omega}
\newcommand{\om}{\omega}
\newcommand{\la}{\lambda}

\newcommand{\te}{\theta}

\newtheorem{thm}{Theorem}[section]

\newtheorem{lem}[thm]{Lemma}
\newtheorem{prop}[thm]{Proposition}
\theoremstyle{definition}
\newtheorem{example}[thm]{Example}
\newtheorem{defn}[thm]{Definition}
\theoremstyle{remark}

\numberwithin{equation}{section}

\newcommand{\CF}{\mathcal{F}}
\newcommand{\CA}{\mathcal{A}}
\newcommand{\CB}{\mathcal{B}}
\newcommand{\CL}{\mathcal{L}}
\newcommand{\CR}{\mathcal{R}}
\newcommand{\CV}{\mathcal{V}}
\newcommand{\CO}{\mathcal{O}}

\begin{document}

%-------------------------------------------------------------------------
% editorial commands: to be inserted by the editorial office
%
%\firstpage{1} \volume{228} \Copyrightyear{2004} \DOI{003-0001}
%\seriesextraline{This is the Concrete Title of this Book\br H.E. R and S.T.C. W, Eds.}
%\firstpage{1}

\title[Twisted Araki--Woods Algebras, the YBE, and QFT]{Twisted Araki--Woods Algebras, the Yang--Baxter Equation, and quantum field theory}
%----------Author 1
\author{Gandalf Lechner}
%\author[Birkh\"auser \textit{et al.}]%
%{Birkh\"{a}user Publishing Ltd.}
%
\address{%
Cauerstr. 11\\
DE 91058 Erlangen\\
Germany}
\email{gandalf.lechner@fau.de}
%
% \thanks{This file has been typeset with the option \texttt{draft} to illustrate that feature and its purpose.}
%----------Author 2
% \author[R.~Ab\l amowicz]{Rafa\l \ Ab\l amowicz}
%\author[]{Rafa\l \ Ab\l amowicz}
% \address{%
% E-i-C of AACA\\
% Sarasota, FL 34238}
% \email{rablamowicz.aaca@birkhauser-science.com}
%
% \keywords{Class file, HTML, journal}
%
% \date{06. November 2023}
\dedicatory{\em 24. November 2023}
%%% ----------------------------------------------------------------------
\begin{abstract}
    This article reviews recent work with Correa da Silva on twisted Araki-Woods algebras, including an introduction to twisted Fock spaces and standard subspaces. We discuss a new family of examples of that framework, coming from the set-theoretic Yang-Baxter equation, and explain the relevance of twisted Araki-Woods algebras in the construction of quantum field theoretic models.
    
    This is a contribution to the proceedings of the XL Workshop on Geometric Methods in Physics 2023 in Białowieża.
\end{abstract}
% \label{page:firstblob}
%%% ----------------------------------------------------------------------
\maketitle
%%% ----------------------------------------------------------------------
% \tableofcontents

\section{Introduction and Overview}

The scope of the Białowieża workshops on Geometric Methods in Physics has been wide throughout the history of this long-running series, including various different areas of mathematics and physics. In line with this approach, the present chapter reviews the recently introduced twisted Araki--Woods algebras~\cite{CorreaDaSilvaLechner:2023} and highlights their connections to various topics, including quantum field theory, operator algebras, free probability, and braided vector spaces.

Twisted Araki--Woods algebras are a family of von Neumann algebras naturally represented on certain twisted Fock spaces $\CF_T(\Hil)$ built on the basis of a Hilbert space $\Hil$ (the single particle space) and a ``twist'', namely a selfadjoint operator $T$ on $\Hil\ot\Hil$ satisfying a subtle positivity condition. The twisted Araki-Woods algebras $\CL_T(H)$ then depend on two data: The twist $T$ and a specific real linear subspace $H\subset\Hil$ (a standard subspace). 

In Section~\ref{sec:twistedfock}, we will review the general formalism of twisted Fock spaces, due to Bożejko and Speicher~\cite{BozejkoSpeicher:1994} and Jørgensen, Schmitt and Werner \cite{JorgensenSchmittWerner:1995}. Depending on the field of application, the twist has various different interpretations, for instance a two-particle interaction in QFT, a deformation of a free group factor in free probability, or the braiding underlying a Nichols algebra. As an original contribution to these proceedings, and to connect to the other talks in the Yang-Baxter session of the workshop, we will in particular explain how set-theoretic solutions to the Yang-Baxter equation fit into this framework, but also discuss  examples from quantum field theory.

The twisted Araki-Woods algebras $\CL_T(H)$ are defined in Section~\ref{sec:alg}, which also contains a concise introduction to standard subspaces for the non-expert reader. Depending on the relative position of $T$ and $H$, we obtain a wide range of von Neumann algebras $\CL_T(H)$. This section reviews our recent results about cyclic and separating Fock vacuum from \cite{CorreaDaSilvaLechner:2023}, deriving the crossing symmetry from elementary particle physics and the Yang-Baxter equation from an operator-algebraic framework. We also give an account of known results on the internal structure of these algebras for special choices of $T$ and/or $H$. This review is accompanied by examples from QFT and set-theoretic solutions to the YBE. In the latter case, we give a new motivation for considering non-degenerate solutions. In Section~\ref{sec:QFT} we sketch applications to constructive quantum field theory which depend on families of twisted Araki-Woods algebras.

\section{Twisted Fock spaces in quantum physics and operator algebras}
\label{sec:twistedfock}

Fock spaces and second quantization procedures appear in many variants in physics and mathematics, linking multi-particle quantum systems, quantum field theory, operator algebras, free probability, and braided vector spaces. In this section we review some of these connections before giving the general definition of twisted Fock space that we will use.

\medskip

In quantum physics, Fock spaces describe a multi-particle system in terms of a corresponding single-particle quantum system by assigning a multi-particle Fock Hilbert space $\CF(\Hil)$ to a single-particle Hilbert space $\Hil$. According to this idea, $\CF(\Hil)$ is defined as a direct sum over ``$n$-particle spaces'', namely certain subspaces of the tensor powers $\Hil\tp{n}$, $n\in\Nl_0$, where $\Hil\tp{0}:=\Cl$. In order to account for distinguishable or indistinguishable particles with Bose/Fermi statistics, one considers different kinds of Fock spaces (including in particular symmetric/antisymmetric/unsymmetrized versions).

What makes Fock spaces mathematically interesting is that in addition to their Hilbert space structure they also have algebraic structure. For example, the unsymmetrized (also called ``full'' or ``Boltzmann'') Fock space $\CF_0(\Hil)=\bigoplus_{n=0}^\infty\Hil\tp{n}$ is the Hilbert space completion of the tensor algebra of $\Hil$, which inherits a ${}^*$-structure from the Hilbert space and acts (from the left) on $\CF_0(\Hil)$ in terms of creation and annihilation operators $a^*(\xi)$, $a(\xi)$, $\xi\in\Hil$, defined by 
\begin{align}
  a^*(\xi)\psi_1\ot\ldots\ot\psi_n := \xi\ot\psi_1\ot\ldots\ot\psi_n,\qquad n\in\Nl,\;\psi_k\in\Hil,
\end{align}
and $a(\xi):=a^*(\xi)^*$. These operators satisfy the Cuntz relations $a(\xi)a^*(\eta)=\langle\xi,\eta\rangle\cdot1$ and generate a $C^*$-algebra closely related to the Cuntz algebra $\CO_{\dim\Hil}$ \cite{Evans:1980,JorgensenSchmittWerner:1994_2}.

In a similar fashion, the Bose/Fermi Fock spaces come automatically with a representation of the CCR/CAR algebras and their respective $C^*$-closures over $\Hil$ \cite{BratteliRobinson:1997}. These operators again act by creation/annihilation operators on the Fock space, but satisfy different commutation relations as a consequence of the symmetry/antisymmetry of the Fock space. Both versions are subsumed in the $q$-deformed relations
\begin{align}\label{eq:qalg}
  a(\xi)a^*(\eta)-q\cdot a^*(\eta)a(\xi) = \langle\xi,\eta\rangle\cdot 1
  ,
\end{align}
which turn into Bosonic/Fermionic relations for $q=1$ and $q=-1$, respectively. Starting from the commutation relations \eqref{eq:qalg} for general parameter $-1\leq q\leq 1$, it is however less clear whether a Hilbert space representation exists. A proof of this fact has been given by Bożejko and Speicher \cite{BozejkoSpeicher:1991} using a $q$-twisted Fock space. This representation can be interpreted either as a generalized Brownian motion or as generalized statistics, and already shows the usefulness of going beyond the usual Bose/Fermi Fock spaces.

The $q$-deformed relations can be significantly generalized by considering quadratic exchange relations of the form
\begin{align}\label{eq:wick}
  a_ia_j^*-\sum_{k,l} T^{kl}_{ij}a_l^*a_k=\delta_{ij}1
\end{align}
(often called ``Wick algebras'' because they allow for a form of normal ordering) and asking for which coefficients $T^{kl}_{ij}$ a Hilbert space representation exists. The idea is that $a_k=a(e_k)$ should correspond to annihilation operators on some Fock space, evaluated on a vector $e_k$ from an orthonormal basis of $\Hil$, and the sum in \eqref{eq:wick} is initially only defined in case it is finite.

Taking into account that on a Fock space, the Fock vacuum vector $\Om=1\oplus0\oplus0\oplus\ldots$ induces the state $\om=\langle\Om,\,\cdot\,\Om\rangle$ on the Wick algebra and the annihilation operators should map $\Om$ to $0$, led Jørgensen, Schmitt and Werner to study Fock-type GNS representations of Wick algebras and derive criteria on the coefficients $T^{kl}_{ij}$ for their positivity \cite{JorgensenSchmittWerner:1995}. This leads to Fock spaces in which the $n$-particle spaces depend on an operator $T$ defining the coefficients $T^{kl}_{ij}$, as will be reviewed below.

Wick relations of a related but different form also appear in the work of the Zamolodchikov brothers \cite{ZamolodchikovZamolodchikov:1979} and Fadeev \cite{Faddeev:1984} on quantum integrable systems. Here the physical idea is to consider creation/annihilation type operators $Z^*(\te),Z(\te)$ representing particles on a spatial line with rapidity $\te\in\Rl$ and obeying relations of the form
\begin{align}\label{eqZ}
  Z(\te)Z^*(\te') = S(\te'-\te)\cdot Z^*(\te')Z(\te)+\delta(\te-\te')\cdot1,
\end{align}
where $S:\Rl\to\Cl$ is a given function satisfying various properties that ensure that it can be interpreted as the elastic two-body S-matrix (here for simplicity taken to be scalar). Such relations are clearly reminiscent of the quadratic Wick relations \eqref{eq:wick}, but have to be understood in terms of  distributions. That is, only some ``smeared'' form of \eqref{eqZ}, i.e. integrated in $\te$ and $\te'$ against test functions, has meaning in terms of actual operators. Due to the factor $S(\te-\te')$, this goes beyond the finite sums in \eqref{eq:wick} when $S$ is not constant. For such scenarios, it is better to define the algebra of interest directly in a Fock representation, which begs the question of how to define a suitable Fock space in the first place \cite{LiguoriMintchev:1995-1,Lechner:2003}. Once constructed, such algebras are of prominent importance in the construction and analysis of integrable quantum field theories \cite{Smirnov:1992,Lechner:2008,BostelmannCadamuro:2012}.

As Nichols algebras, twisted Fock spaces also appear in the context of braided vector spaces\footnote{Many thanks go to Leandro Vendramin for pointing this out to me.}. Here the starting point is a one-particle space $\Hil$ with a braiding, that is a bounded operator $T:\Hil\ot\Hil\to\Hil\ot\Hil$ satisfying the Yang-Baxter equation\footnote{We will use the standard tensor notation $T_k:=1^{\ot (k-1)}\ot T\ot1^{\ot(n-k-1)}\in\B(\Hil\tp{n})$ throughout.}
\begin{align}\label{eq:ybe}
  T_1T_2T_1=T_2T_1T_2.
\end{align}
One then considers the quantum symmetrizer, namely the map, $n\in\Nl$,
\begin{align}\label{eq:quantumsymmetrizer}
  P_{T,n} := \sum_{\pi\in S_n} \rho_{T,n}(\pi) \in \B(\Hil\tp{n}),
\end{align}
where $S_n$ is the symmetric group on $n$ letters with its usual Coxeter generators $\tau_i$, and $\rho_{T,n}(\tau_{i_1}\cdots\tau_{i_l}):=T_{i_1}\cdots T_{i_l}$ is well-defined for any reduced word $\tau_{i_1}\cdots\tau_{i_l}$ by Matsumoto's Theorem \cite{Matsumoto:1964}.

The Nichols algebra \cite{Nichols:1978}, a braided Hopf algebra naturally associated with the braiding~$T$, can explicitly be defined as the quotient vector space\footnote{In the purely algebraic context, $\Hil$ can be an arbitrary vector space, and the algebraic tensor product is used.} $\bigoplus_{n=0}^\infty\Hil\tp{n}/\ker P_{T,n}$, and has various applications in mathematics and physics (see, for example, \cite{Woronowicz:1989,Rosso:1998,Meir:2022}). 

This is essentially the same structure as the twisted Fock spaces introduced below. The focus for us is, however, less on the Hopf algebraic but more on the functional analytic structure. Remarkably, the Nichols algebra is also a pre-Hilbert space in a natural way because the quantum symmetrizers $P_{T,n}$ are positive operators for all $n\in\Nl$ in case $\|T\|\leq1$, as shown by Bożejko and Speicher~\cite{BozejkoSpeicher:1994}.

\bigskip

To set the stage for the following investigations, we now pass to the precise definitions. Throughout the rest of the article, $\Hil$ will denote a complex Hilbert space. Since we want to describe a family of Fock spaces over $\Hil$ that includes all the scenarios mentioned above (and many more), we will need a form of the quantum symmetrizer that can be defined without requiring the Yang-Baxter equation~\eqref{eq:ybe}. Similar to ~\cite{BozejkoSpeicher:1994}, given an operator $T\in\B(\Hil\ot\Hil)$ we define $P_{T,n}\in\B(\Hil\tp{n})$ inductively by
\begin{align}\label{eq:PTn}
  P_{T,1}=1,\qquad P_{T,n+1}=(1\ot P_{T,n})(1+T_1+T_1T_2+\ldots+T_1\cdots T_{n}).
\end{align}
Note that in case $T$ satisfies the Yang-Baxter equation, $P_{T,n}$ coincides with the quantum symmetrizer. In that case, we also have the alternative recursion relation 
\begin{align}\label{rightP}
 P_{T,n+1}=(P_{T,n}\ot 1)(1+T_n+T_nT_{n-1}+\ldots+T_n\cdots T_{1}).
\end{align}
In general, however, we have to make a choice and we here choose the ``left'' version \eqref{eq:PTn}.

\begin{defn}\label{def:twist}
	A {\em twist} is a selfadjoint operator in ${\mathcal B}(\Hil\ot\Hil)$ such that $\|T\|\leq1$ and $P_{T,n}\geq0$ for all $n\in\Nl$. A twist is called {\em strict} if $\ker P_{T,n}=\{0\}$ for all $n\in\Nl$.
\end{defn}

Given any twist $T$, we can now introduce the new scalar products $\langle\,\cdot\,,\,\cdot\,\rangle_T := \langle\,\cdot\,,P_{T,n}\,\cdot\,\rangle$ on $\Hil\tp{n}/\ker P_{T,n}$. This constitutes the definition of a twisted Fock space:

\begin{defn}
  Let $\Hil$ be a Hilbert space and $T$ a twist. The {\em twisted Fock space} is 
  \begin{align}
      \CF_T(\Hil)
      =
      \bigoplus_{n=0}^\infty \overline{\Hil\tp{n}/\ker P_{T,n}},
  \end{align}
  where the bar indicates completion w.r.t. the norm induced by the scalar product $\langle\,\cdot\,,\,\cdot\,\rangle_T = \langle\,\cdot\,,P_{T,n}\,\cdot\,\rangle$.
\end{defn}
The family of twisted Fock spaces includes all familiar types of Fock spaces. For example, the zero operator $T=0$ is easily seen to be a strict twist, with $\CF_0(\Hil)$ equal to the full Fock space over $\Hil$. As another example, consider $T=F:v\ot w\mapsto w\ot v$, the tensor flip on $\Hil\ot\Hil$. In this case, one can check that $\frac{1}{n!}P_{T,n}$ coincides with the orthogonal projection onto the symmetric subspace of $\Hil\tp{n}$, so that we get an identification with the Bosonic Fock space over $\Hil$ \cite{CorreaDaSilvaLechner:2023}. Similarly, $T=-F$ corresponds to the Fermi Fock space, and $T=qF$ to the $q$-twisted Fock space mentioned before.

In general, it is not straightforward to check whether a given operator~$T$ is a twist. However, some sufficient conditions are known, which we now summarize. Parts a) and b) are due to Jørgensen, Schmitt and Werner \cite{JorgensenSchmittWerner:1995}, and part c) is due to Bożejko and Speicher \cite{BozejkoSpeicher:1994}.

\begin{thm}\label{theorem:T}
	Let $T=T^*\in\CB(\Hil\ot\Hil)$.
	
\renewcommand{\labelenumi}{\normalfont\alph{enumi})}
	\begin{enumerate}
		\item\label{item:Tsmallnorm} If $\|T\|\leq\frac{1}{2}$, then $T$ is a strict twist.
		\item\label{item:Tpositive} If $T\geq0$, then $T$ is a strict twist.
		\item\label{item:Tbraided} If $\|T\|\leq1$ and $T$ satisfies the Yang-Baxter equation, i.e.
		\begin{align}
			T_1T_2T_1=T_2T_1T_2,
		\end{align}
		then $T$ is a twist. In case $\|T\|<1$, this twist is strict.
	\end{enumerate}
\end{thm}

Many examples of twists are discussed in \cite{JorgensenProskurinSamoilenko:2001,DaletskiiKalyuzhnyLytvynovProskurin:2020,CorreaDaSilvaLechner:2023}. We restrict ourselves here to present two families of examples of braided twists, i.e. twists satisfying the assumptions of Theorem~\ref{theorem:T}~c). The first class of examples connects to the other talks in the Yang-Baxter session of the workshop, and the second class is connected to applications in QFT.

\begin{example}{\bf(Set-theoretic solutions to the YBE)}\label{ex:setsols}
A set-theoretic solution of the YBE consists of a set $X$ and a map $r:X^2\to X^2$ such that $r_1r_2r_1=r_2r_1r_2$ as maps $X^3\to X^3$ in standard leg notation\footnote{That is, $r_1(x,y,z)=(r(x,y),z)$ and $r_2(x,y,z)=(x,r(y,z))$.}. Often set-theoretic solutions are written as
\begin{align}\label{eq:setYBE}
  r(x,y) = (\la_x(y),\rho_y(x)),\qquad x,y\in X,
\end{align}
with maps $\la_x,\rho_y:X\to X$. The Yang-Baxter equation can be rewritten as a set of three equations for $\la$ and $\rho$ by straightforward computation, but we will not need these here. For simplicity, we will restrict to finite sets $|X|<\infty$, but many parts of the subsequent analysis easily generalize to infinite sets.

Any set-theoretic solution can be linearized. To this end, we consider the vector space $\Hil=\ell_\Cl^2(X)$, the complex Hilbert space of square summable sequences indexed by $X$, and may naturally identify $X$ with an orthonormal basis by $x\mapsto\delta_x$. The linearisation of $r$ is defined to be the unique linear bounded operator $T:\Hil\ot\Hil\to\Hil\ot\Hil$ given by
\begin{align}
  T(x\otimes y) = (\la_x(y)\otimes\rho_y(x)),\qquad x,y\in X.
\end{align}
By definition, $T$ satisfies the linear YBE \eqref{eq:ybe}. This observation by Drinfeld \cite{Drinfeld:1992} originally motivated the study of set-theoretic solutions to the YBE, which is now a research field in its own right (see, for example, \cite{BrzezinskiColazzoDoikouVendramin:2023}).

The operator $T$ is a candidate for a twist operator according to Theorem~\ref{theorem:T}~c). However, for~$T$ to be a twist it also needs to be selfadjoint and of norm $\|T\|\leq1$. Both these properties are not automatically satisfied.

\begin{lem}
  Let $(X,r)$ be a set-theoretic solution of the YBE, and $(\Hil,T)$ its linearization. Then $T$ is selfadjoint if and only if $T$ is involutive, i.e. $T=T^{-1}$. In this case, $\|T\|=1$. Such a twist is strict if and only if $r=\id_{X^2}$.
\end{lem}
\begin{proof}
  If $T=T^*$ is selfadjoint, we have by definition of $T$ and $\Hil$ for any $x,y,x',y'\in X$
  \begin{align*}
    \delta_{(x',y'),r(x,y)}
    =
    \langle x'\ot y',T(x\ot y)\rangle
    =
    \langle T(x'\ot y'),x\ot y\rangle
    =
    \delta_{r(x',y'),(x,y)}
    .
  \end{align*}
  This immediately implies $r^{-1}(\{(x,y)\})=\{r(x,y)\}$ for any $(x,y)\in X^2$, so $r$ is bijective with $r^{-1}=r$. Hence also the linearization $T$ is invertible and involutive.

  If, on the other hand, $T=T^{-1}$ is involutive, it maps the orthonormal basis $\{x\ot y\,:\,x,y\in X\}$ onto itself. Hence $T$ is unitary, and in view of $T=T^{-1}=T^*$ also selfadjoint. Clearly unitary solutions have norm $\|T\|=1$.

  A unitary involutive solution $T$ of the YBE generates a unitary representation $\rho_{T,n}$ of the symmetric group $S_n$ on $\Hil\tp{n}$ by sending the transposition $\tau_i\in S_n$ to $T_i\in\CB(\Hil\tp{n})$. This implies that up to a factor $n!$, the quantum symmetrizer \eqref{eq:quantumsymmetrizer} coincides with the orthogonal projection onto the subspace of $\Hil\tp{n}$ consisting of all vectors invariant under this representation. Hence $T$ being strict, namely $\ker P_{T,n}=\{0\}$, is equivalent to $P_{T,n}=n!$, which in turn is equivalent to $T=1$, i.e. $r=\id_{X^2}$.
\end{proof}

Involutive solutions to the set-theoretic YBE are a subject of current research in the set-theoretic setting \cite{EtingofSchedlerSoloviev:1999,Gateva-IvanovaVandenBergh:1998,Rump:2007,CedoJespersOkninski:2014}. Up to a natural equivalence given by the $S_\infty$-representations they generate, they have been classified in \cite{LechnerPennigWood:2019}.
\end{example}

\begin{example}{\bf(Solutions to the YBE with spectral parameter \cite{CorreaDaSilvaLechner:2023})}\label{ex:specpar}
  The second class of examples arises from the Yang-Baxter equation with spectral parameter, as they appear in integrable quantum field theories \cite{AbdallaAbdallaRothe:2001}. We consider a finite-dimensional complex Hilbert space $V$ and the one-particle space $\Hil:=L^2(\Rl\to V,d\te)\cong L^2(\Rl,d\te)\ot V$. Given any measurable function $S:\Rl\to\B(V\ot V)$ bounded by $\|S\|_\infty\leq 1$ that satisfies the YBE with spectral parameter, namely
  \begin{align*}
    S(\te)_1S(\te+\te')_2S(\te')_1
    =
    S(\te')_2S(\te+\te')_1S(\te)_2,\qquad \te,\te'\in\Rl,
  \end{align*}
  and $S(\te)^*=S(-\te)$, we consider
  \begin{align}\label{eq:TS}
    T_S:\Hil\ot\Hil\to\Hil\ot\Hil,\qquad 
    (T_S\psi)(\te_1,\te_2):=S(\te_2-\te_1)\psi(\te_2,\te_1).
  \end{align}
  Thanks to the properties of $S$, this is a selfadjoint operator of norm $\|T\|=\|S\|_\infty\leq1$ satisfying the YBE on $\Hil\tp{3}$, i.e. a twist. 
  
  In the context of QFT, $S$ models a relativistic elastic two-particle interaction, and the Yang-Baxter equation is a consistency condition to allow for a consistent factorization of a three-particle scattering process into three two-particle processes. For concrete examples, see \cite{AbdallaAbdallaRothe:2001}.
  \hfill$\square$
\end{example}

On any twisted Fock space $\CF_T(\Hil)$, we have a natural unital ${}^*$-algebra of creation/annihilation type  operators. Denoting the quotient map $\Hil\tp{n}\to\Hil\tp{n}/\ker P_{T,n}$ by $[\cdot]$, we set
\begin{align}
  a_{L,T}^\star(\xi)[\Psi_n] := [\xi\ot\Psi_n],\qquad \xi\in\Hil,\;\Psi_n\in\Hil\tp{n},
\end{align}
and extend to a densely defined operator in $\CF_T(\Hil)$ by linearity. Here ``$L$'' reminds us that we are working with the ``left'' version of the $P_{T,n}$. The star ${}^\star$ will always be used to indicate adjoints w.r.t. the $T$-dependent scalar product of $\CF_T(\Hil)$, and we write $a_{L,T}(\xi):=a_{L,T}^\star(\xi)^\star$ as usual.

As we shall see below, various properties of the operators $a^\#_{L,T}(\xi)$ (for operators $A$, we write $A^\#$ to denote either $A$ or $A^\star$) and certain von Neumann algebras generated by them differ sharply depending on whether we have $\|T\|=1$ or $\|T\|<1$. A first indication of these two regimes is that in the braided case, $T$ is strict for $\|T\|<1$ (Theorem~\ref{theorem:T}~c)). In general, a useful intuition to have is that in the extreme case $T=0$, we are presented with an ``extremely noncommutative'' free algebra (tensor algebra). The case $\|T\|=1$, on the other hand, includes in particular $T=F$ (the tensor flip). This yields the CCR algebra and corresponding local quantum field theories, which are intuitively speaking ``much more commutative''.

On a technical level, an important difference is that $a^\#_{L,T}(\xi)$ is bounded for $\|T\|<1$ \cite{BozejkoSpeicher:1994}, but unbounded for $\|T\|=1$ (unless $T=-F$, where the CAR relations imply boundedness).

\section{Localized von Neumann algebras and standard subspaces}
\label{sec:alg}

Given any twist $T$, the twisted Fock space $\CF_T(\Hil)$ construction provides us with the unital ${}^*$-algebra generated by $a_{L,T}^\#(\xi)$, $\xi\in\Hil$ and the Fock vacuum vector $\Om\in\CF_T(\Hil)$. In the following we want to use these operators to generate certain von Neumann algebras, generically denoted $\CA\subset\CB(\CF_T(\Hil))$ for the time being, such that $\Om$ is cyclic (meaning that $\CA\Om\subset\CF_T(\Hil)$ is dense) and separating (meaning that $\CA\ni A\mapsto A\Om\in\CF_T(\Hil)$ is injective).

A pair $(\CA,\Om)$ consisting of a von Neumann algebra with cyclic separating vector is the starting point of Tomita-Takesaki modular theory \cite{Takesaki:2003} and also of central importance in algebraic quantum field theory, where these properties are consequences of the basic principles of Einstein locality and positivity of the energy \cite{Haag:1996}.

As $a_{L,T}(\xi)\Om=0$, the Fock vacuum $\Om$ does not separate any algebra containing $a_{L,T}(\xi)$, $\xi\neq0$. We therefore consider the Segal type field operator
\begin{align}
  \phi_{L,T}(\xi) := a_{L,T}^\star(\xi)+a_{L,T}(\xi),\qquad\xi\in\Hil.
\end{align}
This operator is selfadjoint and bounded for $\|T\|<1$ and essentially selfadjoint on the subspace of finite particle vectors for $\|T\|=1$.

The von Neumann algebras that we want to study are of the form
\begin{align}\label{eq:LTH}
  \CL_T(H) := \{\phi_{L,T}(h)\,:\,h\in H\}''.
\end{align}
Since $\xi\mapsto\phi_{L,T}(\xi)$ is {\em real} linear because $a_{L,T}^\star(\xi)$ depends linearly, but $a_{L,T}(\xi)$ depends antilinearly on $\xi$, the set $H\subset\Hil$ can be taken to be a real linear subspace. In view of $a_{L,T}(h)=\frac12(\phi_{L,T}(h)+i\phi_{L,T}(ih))$, we have (at least) to choose $H$ in such a way that $H\cap iH=\{0\}$, otherwise $\Om$ will not be separating for $\CL_T(H)$. On the other hand, $\CL_T(H)\Om$ contains all the one-particle vectors $h_1+ih_2$, with $h_1,h_2\in H$, so that we are led to require that $H+iH\subset\Hil$ is dense to ensure cyclicity. We will therefore consider standard subspaces.

\begin{defn}
  A {\em standard subspace} is a closed real linear subspace $H\subset\Hil$ such that $H+iH$ is dense in $\Hil$ and $H\cap iH=\{0\}$.
\end{defn}

Given a twist $T\in\CB(\Hil\ot\Hil)$ and a standard subspace $H\subset\Hil$, the associated {\em twisted Araki-Woods algebra} is the von Neumann algebra defined in \eqref{eq:LTH}.

\subsection{Standard subspaces}

Simple examples of standard subspaces are $\Rl^n$ as a real subspace of $\Cl^n$, or the real-valued functions in $L^2(\Rl^n\to\Rl)$ as a real subspace of $L^2(\Rl\to\Cl)$. Slightly more generally, one may consider an orthonormal basis $(e_n)_{n\in\Nl}$ of a Hilbert space $\Hil$ and define  $H$ as the closure of the real linear span of this orthonormal basis, which clearly is a standard subspace. Such standard subspaces are called {\em maximally abelian} in the literature \cite{Longo:2008}. We mention as an aside that for twist $T=0$ and a maximally abelian standard subspace $H$, the twisted Araki-Woods algebra $\CL_0(H)$ is isomorphic to the group von Neumann algebra of the free group on $\dim\Hil$ generators \cite{Voiculescu:1985}. This explains the relevance of the algebras $\CL_T(H)$ as deformations of free group factors in free probability \cite{Shlyakhtenko:1997,Hiai:2003}.

Maximally abelian standard subspaces are however very special examples. To understand the general structure of standard subspaces, we consider the {\em Tomita operator} $S_H$ of $H\subset\Hil$, that is the map
\begin{align*}
  S_H:H+iH\to\Hil,\qquad h_1+ih_2\mapsto h_1-ih_2.
\end{align*}
This is a closed antilinear involution, and its polar decomposition $S_H=J_H\Delta_H^{1/2}$ consists of an antiunitary involution $J_H$ (modular conjugation) and a strictly positive typically unbounded operator $\Delta_H^{1/2}>0$ (modular operator) satisfying the {\em modular condition}
\begin{align}\label{eq:JD}
  J_H\Delta_H^{1/2}J_H=\Delta_H^{-1/2}.
\end{align}
Any pair $J_H,\Delta_H^{1/2}$ of operators satisfying these conditions defines a unique standard subspace $H=\ker(J_H\Delta_H^{1/2}-1)$. The aforementioned maximally abelian standard subspaces are characterized by $\Delta_H^{1/2}=1$ \cite{Longo:2008}.

Standard subspaces come in pairs: With a standard subspace $H$, also its symplectic complement $H':=\{\psi\in\Hil\,:\,{\rm Im}\langle\psi,h\rangle=0\,\forall h\in H\}$ is a standard subspace, and $H''=H$. A basic fact about standard subspaces is a variant of Tomita's Theorem, expressing that the modular unitaries act as automorphisms of $H$, and $J_H$ exchanges $H$ and $H'$:
\begin{align}
  \Delta_H^{it}H=H,\qquad t\in\Rl,\qquad\qquad J_HH=H'.
\end{align}

Given a von Neumann algebra $\M$ on some Hilbert space $\Hil$ with a vector $\Om$, the set $H:=\{A\Om\,:\,A=A^*\in\M\}^-$ is a standard subspace if and only if $\Om$ is cyclic and separating for $\M$. Thus standard subspaces appear very naturally in the context of von Neumann algebras, and the symplectic complement $H'$ of a standard subspace plays the role of the commutant $\M'$ of a von Neumann algebra $\M$.

In quantum field theory, standard subspaces can be used to encode localization regions. We restrict ourselves to an example from Minkowski space. 

\begin{example}[\bf QFT examples of standard subspaces]\label{example:bw}
  Consider the test function space $C_c^\infty(\Rl^d)$ on $d$-dimensional Minkowski space, $d\geq1+1$, and $\CO\subset\Rl^d$ a localization region (a set with interior points such that its causal complement $\CO'$ also has interior points).
  
  Fixing a mass parameter $m>0$, we consider the Hilbert space $\Hil=L^2(\Rl^{d-1},(\|p\|^2+m^2)^{-1/2}dp)$ carrying the usual spin zero mass $m$ irreducible positive energy representation of the Poincaré group. We then use the map 
  \begin{align*}
    C_c^\infty(\Rl^d) \ni f \longmapsto f^+ \in\Hil,\qquad f^+(p):=\tilde f(\sqrt{\|p\|^2+m^2},p)
  \end{align*}
  (with $\tilde f$ the Fourier transform of $f$), associating test functions with single particle states (or solutions to the Klein-Gordon equation), and
  \begin{align}
    H(\CO) := \{f^+\,:\,\supp f\subset\CO\}^- \subset \Hil.
  \end{align}
  Then $H(\CO)$ is a standard subspace. This is a consequence of the Reeh-Schlieder property of the vacuum \cite{ReehSchlieder:1961} and can be proven as a special case of a one-particle version of the Reeh-Schlieder Theorem \cite{StreaterWightman:1964}.
  
  The physical interpretation is that elements of $H(\CO)$ are localized in the spacetime region $\CO$ in the sense of being excitations of the vacuum by observables localized in $\CO$. From this point of view, standard subspaces can be seen as an abstract notion of localization region.
  
  In most cases, an explicit description of the modular data $J_{H(\CO)}, \Delta_{H(\CO)}$ of $H(\CO)$ is not known. The most prominent case in which these operators are known and act geometrically \cite{BuchholzDreyerFlorigSummers:2000} is the case where the region $\CO$ is the Rindler wedge
  \begin{align*}
   \CO=W=\{x=(x_0,x_1,\ldots,x_{d-1})\,:\,x_1>|x_0|\},
   \end{align*}
  or a Poincaré transform thereof. 
  
  In this case the modular unitaries $\Delta_{H(W)}^{it}$ act as Lorentz boosts in the $x_1$-direction (Bisognano-Wichmann Theorem \cite{BisognanoWichmann:1975,BisognanoWichmann:1976}). For recent generalizations of Bisognano-Wichmann results to representations of Lie groups on homogeneous spaces, see \cite{MorinelliNeeb:2021,MorinelliNeebOlafsson:2022_2}.
  
  For later reference, we mention that in the case of dimension $d=2$, the wedge standard subspace can be reformulated as follows. Changing variables from $p$ to $p=m\sinh\te$, one finds that in the Hilbert space $\Hil=L^2(\Rl,d\te)$, the operators
  \begin{align*}\label{hardyjd}
    (\Delta^{it}\psi)(\te)=\psi(\te-2\pi t),\qquad (J_H\psi)(\te)=\overline{\psi(\te)}
  \end{align*}
  define a standard subspace $H$. Concretely, $H$ is given by those $L^2$-functions~$h$ that have an analytic  continuation to functions in the Hardy space $H^2({\mathbb S}_\pi)$ on the strip ${\mathbb S}_\pi=\{\te\in\Cl\,:\,0<{\rm Im}\te<\pi\}$ and satisfy the symmetry condition $\overline{h(\te+i\pi)}=h(\te)$, $\te\in\Rl$ \cite{LechnerLongo:2014}.
\end{example}

\subsection{Crossing Symmetry and Yang-Baxter Equation}

Fixing a standard subspace $H\subset\Hil$ and a compatible twist $T$, we now review known results about the twisted Araki-Woods algebra $\CL_T(H)$ in an abstract setting. We begin with the question when the Fock vacuum $\Om$ is cyclic and separating as a basic prerequisite for both, modular theory and applications in QFT.

To this end, we call the pair $(H,T)$ {\em compatible} if
\begin{align}
  [T,\Delta_H^{it}\ot\Delta_H^{it}]=0,\qquad t\in\Rl.
\end{align}
Under this basic symmetry requirement, it turns out that $\Om$ is separating for~$\CL_T(H)$ (it is always cyclic) if and only if two conditions are satisfied: The Yang-Baxter equation and a crossing symmetry condition, defined as follows.

\begin{defn}{\bf \cite{CorreaDaSilvaLechner:2023} (crossing symmetry)}\label{def:crossing}\\
	Let $H\subset\Hil$ be a standard subspace. A bounded operator $T\in\B(\Hil\ot\Hil)$ is called {\em crossing-symmetric w.r.t. $H$} if for all $\psi_1,\ldots,\psi_4\in \Hil$, the function
	\begin{align}\label{eq:crossing-function}
		T^{\psi_2,\psi_1}_{\psi_3,\psi_4}(t)
		:=
		\langle\psi_2\otimes \psi_1,\,(\Delta_H^{it}\otimes 1)T(1\otimes \Delta_H^{-it})(\psi_3\otimes \psi_4)\rangle
	\end{align}
	has a continuous and bounded extension to the strip $\Rl+i[0,\frac12]$ which is analytic in $\Rl+i(0,\frac12)$, and satisfies the boundary condition, $t\in\Rl$,
	\begin{equation}\label{eq:Crossingboundary}
% 	\begin{aligned}
		T^{\psi_2,\psi_1}_{\psi_3,\psi_4}(t+\tfrac{i}{2})
		=
		\langle \psi_1\otimes J_H\psi_4,\,(1\otimes \Delta_H^{it})T(\Delta_H^{-it}\otimes 1)(J_H\psi_2\otimes \psi_3)\rangle
% 		\\
% 		&=\overline{(T^*)^{J_H \psi_2,\psi_3}_{\psi_1,J_H\psi_4}(t)}
		.
% 	\end{aligned}
	\end{equation}
\end{defn}

It is straightforward to check that multiples of the flip, i.e. $T=qF$, $-1\leq q\leq 1$, satisfy crossing symmetry (the functions $T^{\psi_2,\psi_1}_{\psi_3,\psi_4}$ are constant in this case), and multiples of the identity, i.e. $T=q1$, violate crossing symmetry. In general, crossing symmetry is a subtle condition which asks that a) the analytic continuation mentioned above exists, b) the boundary values at $\Rl+\frac{i}{2}$ are given by a bounded operator (not just a quadratic form), and c) the symmetry condition \eqref{eq:Crossingboundary} holds.

Crossing symmetry is an abstraction of the crossing symmetry of scattering of elementary particles (relating the scattering of particles and antiparticles), and reminiscent of the KMS condition in the description given above. This is an indication that it is related to separating properties, and indeed the following theorem holds:

\begin{thm}{\bf \cite{CorreaDaSilvaLechner:2023}}\label{thm1}
  Let $H\subset\Hil$ be a standard subspace and $T$ a compatible twist. Then $\Om$ is cyclic and separating for $\CL_T(H)$ if and only if $T$ is crossing symmetric w.r.t. $H$ and satisfies the Yang-Baxter Equation.
\end{thm}

Whereas usually the Yang-Baxter equation and crossing symmetry are {\em assumed} for building certain models, this theorem {\em derives} both these structures from a general operator-algebraic framework. It might also explain why out of the very many existing twists (see Theorem~\ref{theorem:T}), the braided twists have received most attention.

Once the conditions in Theorem~\ref{thm1} are satisfied, $\Om$ is cyclic and separating for $\CL_T(H)$ and hence $\{A\Om\,:\,A=A^*\in\CL_T(H)\}^-$ is a standard subspace in $\CF_T(\Hil)$, defining modular data $J_{T,H}$, $\Delta_{T,H}$. These are related to $J_H,\Delta_H$ by 
\begin{align}
  J_{T,H}\psi_1\ot\ldots\ot\psi_n &= J_H\psi_n\ot\ldots\ot J_H\psi_1,
  \\
  \Delta^{it}_{T,H}\psi_1\ot\ldots\ot\psi_n &= \Delta^{it}_H\psi_1\ot\ldots\ot \Delta^{it}_H\psi_n,
\end{align}
for $\psi_i\in\Hil$, $n\in\Nl$. We refer to \cite{CorreaDaSilvaLechner:2023} for details.

The modular conjugation $J_{T,H}$ also allows us to determine the commutant $\CL_T(H)'$ of $\CL_T(H)$. It is given by a ``right'' version of the ``left'' von Neumann algebra $\CL_T(H')$: Recall that in the initial construction of the twisted Fock space, we had a choice between a left and right version \eqref{rightP} for the definition of the kernels $P_{T,n}$. In case $T$ satisfies the YBE, both versions agree. In that case, $\CF_T(\Hil)$ also carries creation/annihilation type operators that act from the right instead of the left. We are therefore in a position to also consider the corresponding right versions $\CR_T(H)$ of $\CL_T(H)$. The modular conjugation then implements a duality between these two, namely 
\begin{align}
  \CL_T(H)' = J_{T,H}\CL_T(H)J_{T,H} = \CR_T(H').
\end{align}
Again, we refer to \cite{CorreaDaSilvaLechner:2023} for a more detailed discussion and proofs.

We now revisit Example~\ref{ex:setsols} and Example~\ref{ex:specpar} in connection with crossing symmetry.

\begin{example}{\bf (Crossing-symmetric set-theoretic solutions of the YBE)}\\
  This example is a continuation of Example~\ref{ex:setsols} about set-theoretic solutions to the YBE. In that example, we considered set theoretic solutions $r:X^2\to X^2$ without any additional requirements, but found that for them to linearize to twists, $r$ has to be bijective and indeed involutive.
  
  Set-theoretic solutions are often studied under the additional assumption that they are {\em non-degenerate}, namely that the left and right projections $\la_x$, $\rho_x$ of $r$ (cf.~\ref{eq:setYBE}) are bijections of $X$ for any $x$. While this assumption is not satisfied in general (a trivial counterexample is $r=\id_{X^2}$), it does allow to use powerful tools from group theory and algebra (braces) when it is available \cite{Rump:2007,CedoJespersOkninski:2014,GuarnieriVendramin:2017}.

  Note that to the linearization $(\Hil,T)$ of a set-theoretic solution $(X,r)$, we may naturally associate a standard subspace. Namely, we pick an involutive bijection $j:X\to X$. Then the antilinear extension of $j$ is an antiunitary operator $J_H$ on $\Hil$, and setting $\Delta_H:=1$ we obtain a (maximally abelian) standard subspace $H=\ker(1-J_H)$ with modular data $J_H$ and $\Delta_H=1$, which is trivially compatible with $T$.

\begin{prop}
  Let $(X,r)$ be a set-theoretic solution to the YBE such that its linearization $(\Hil,T)$ is crossing symmetric w.r.t. a standard subspace $H$ of the form described above. Then $r$ is non-degenerate and 
  \begin{align}
    \rho_x = j\la_x^{-1}j,\qquad x\in X.
  \end{align}
\end{prop}
\begin{proof}
  Let $x_1,\ldots,x_4\in X$. By crossing symmetry, the constant function 
  \begin{align*}
    f(t)
    &:=
    \langle x_2\ot x_1,(\Delta_H^{it}\ot 1)T(1\ot\Delta_H^{-it})(x_3\ot x_4)\rangle
    \\
    &=
    \delta_{x_2,\la_{x_3}(x_4)}\delta_{x_1,\rho_{x_4}(x_3)}
  \end{align*}
  must analytically continue to 
  \begin{align*}
    f(t+\tfrac{i}{2})
    &:=
    \langle x_1\ot J_Hx_4,(1\ot\Delta_H^{it})T(\Delta_H^{-it}\ot1)(J_Hx_2\ot x_3)\rangle
    \\
    &=
    \delta_{x_1,\la_{j(x_2)}(x_3)}\delta_{j(x_4),\rho_{x_3}(j(x_2))},
  \end{align*}
  i.e. we obtain the condition $\delta_{x_2,\la_{x_3}(x_4)}\delta_{x_1,\rho_{x_4}(x_3)}=\delta_{x_1,\la_{j(x_2)}(x_3)}\delta_{j(x_4),\rho_{x_3}(j(x_2))}$.
  
  Setting $x_1:=\rho_{x_4}(x_3)$ and $x_2:=\la_{x_3}(x_4)$ yields $\id_X=\rho_x j\la_x j$, and setting $x_1:=\la_{j(x_2)}(x_3)$ and $x_4:=j\rho_{x_3}(j(x_2))$ yields $\id_X=\la_xj\rho_x j$ for all $x\in X$. These equations clearly imply that both $\la_x$ and $\rho_x$ are bijections satisfying $\rho_x=j\la_x^{-1}j$. In particular, $r$ is non-degenerate.
%   
%   Furthermore, the symmetry condition \eqref{..} further constrains $r$. In terms of the flip $\tau(x,y)=(y,x)$, we have 
%   \begin{align*}
%     r^{-1}(x,y)
%     =
%     (j\times j)\tau r\tau(j\times j)(x,y)
%     =
%     (j\rho_{j(x)}j(y),j\la_{j(y)}j(x))
%   \end{align*}
\end{proof}

This observation can be seen as another motivation or derivation of non-degeneracy from crossing symmetry.

Examples of crossing-symmetric set-theoretic solutions as in this proposition are permutation solutions, namely maps $r(x,y)=(\pi(y),\pi^{-1}(x))$, with $\pi:X\to X$ a bijection commuting with $j$.
\end{example}

\begin{example}{\bf (Crossing-symmetric solutions of the YBE with spectral parameter \cite{CorreaDaSilvaLechner:2023})}\label{example:crossqft}\\
  We now revisit Example~\ref{ex:specpar}, which was built on the vector-valued $L^2$-space $\Hil=L^2(\Rl\to V)=L^2(\Rl,d\te)\ot V$, with $V$ a finite-dimensional Hilbert space. We describe a standard subspace of tensor product form $H=H_0\ot L$, where $H_0\subset L^2(\Rl,d\te)$ and $L\subset V$ are both standard subspaces. 
  
  For $H_0$, we take the standard subspace described in \eqref{hardyjd}, and for $L$, we take a maximally abelian standard subspace, as in the previous example. Concretely, this amounts to  the modular data
  \begin{align}
    (\Delta_H^{it}\psi)(\te)=\psi(\te-2\pi t),
    \qquad 
    (J_H\psi)(\te)=J_L\psi(\te),
  \end{align}
  where $J_L$ is an antiunitary involution on $V$. The underlying standard subspace consists of all elements $h$ of the vector-valued Hardy space $H^2({\mathbb S}_\pi)\ot V$ satisfying $h(\te+i\pi)=J_Lh(\te)$.

  The twists $T_S$ \eqref{eq:TS} considered in Example~\ref{ex:specpar} are then automatically compatible with $H$ because the function $S$ only depends on differences of the variables $\te_1,\te_2$. Crossing symmetry is satisfied when $S:\Rl\to\B(V\ot V)$ has a holomorphic and bounded extension to the strip ${\mathbb S}_\pi$, with the boundary values satisfying
  \begin{align*}
    \langle v_2\ot v_1,S(t+i\pi)\,v_3\ot v_4\rangle 
    =
    \langle v_1\ot jv_4,S(-t)\,jv_2\ot v_3\rangle 
    ,\qquad
    t\in\Rl,
  \end{align*}
  for all $v_1,\ldots,v_4\in V$.
  
  In this setting, our abstract crossing symmetry coincides with the crossing symmetry from scattering theory, and specifically with crossing symmetry in integrable quantum field theories. Various examples of functions $S$ satisfying crossing and the Yang-Baxter equation with spectral parameter are known, although a complete classification has not been reached yet.
  
  In the case of scalar particles, given by $V=\Cl$, the Yang-Baxter equation becomes trivial. If one then also asks $T_S$ to be unitary, the possible functions $S$ are exactly the inner functions on the strip $0<{\rm Im}\te<\pi$ that satisfy the two symmetry conditions $S(-\te)=\overline{S(\te)}=S(\te+i\pi)$, $\te\in\Rl$ \cite{Lechner:2008}.
  
  For some specific examples for $\dim V>1$, see \cite{AbdallaAbdallaRothe:2001, AlazzawiLechner:2016}.
\end{example}

\subsection{The internal structure of twisted Araki-Woods algebras}

While the results in the previous section clarified under which conditions on $(T,H)$, the Fock vacuum is cyclic and separating for the twisted Araki-Woods algebra $\CL_T(H)$, they do not address the internal structure of these algebras. 

The case of the twist $T=qF$, with $-1<q<1$, has been considered in most detail in the literature. Note that this twist is automatically compatible with any standard subspace, and the Yang-Baxter equation and crossing symmetry are satisfied. 

In that case, the structure of $\CL_T(H)$ is well understood:

\begin{thm}{\bf \cite{KumarSkalskiWasilewski:2023}}\label{thm:inn1}
  Let $-1<q<1$ and let $H\subset\Hil$ be an arbitrary standard subspace with $\dim H\geq2$. Then $\CL_{qF}(H)$ is a non-injective factor of type 
  \begin{align*}
    \begin{cases}
      \text{\em III}_1 & \text{if}\; G=\Rl_*^\times\\
      \text{\em II}_\la & \text{if}\; G=\la^\Zl,\;0<\la<1\\
      \text{\em II}_1 & \text{if}\; G=\{1\}
    \end{cases},
  \end{align*}
  where $G\subset\Rl_*^\times$ is the closed subgroup generated by the spectrum of $\Delta_H$. If $\dim H<\infty$, then these factors are solid and full.
\end{thm}

This recent theorem of Kumar, Skalski and Wasilewski settled in particular the long-standing question of factoriality of $\CL_{qF}(H)$ for all $q$ and all $H$. It builds on important previous work by many authors, including in particular Miyagawa and Speicher \cite{MiyagawaSpeicher:2023} and Nelson \cite{Nelson:2015}. We refer to \cite{KumarSkalskiWasilewski:2023} for a detailed description and references regarding the history of the factoriality problem of the $q$-twisted Araki-Woods factors.

Even more recently, Yang has generalized these methods to more general twists \cite{Yang:2023}, namely arbitrary compatible braided crossing-symmetric twists on finite dimensional spaces:

\begin{thm}{\bf \cite{Yang:2023}}
  Let $H\subset\Hil$ be a finite-dimensional standard subspace.
  \renewcommand{\labelenumi}{\normalfont\alph{enumi})}
  \begin{enumerate}
   \item  Let $T$ be a compatible braided crossing-symmetric twist with $\|T\|<1$. Then $\CL_T(H)$ is a factor. The type of this factor is determined by the closed subgroup $G\subset\Rl_*^\times$  generated by the spectrum of $\Delta_H$ exactly as in~Thm.~\ref{thm:inn1}.
   \item There exists a constant $q_H>0$ such that for any compatible braided crossing symmetric twist with $\|T\|<q_H$, the twisted Araki-Woods algebra $\CL_T(H)$ is isomorphic to the free Araki-Woods algebra $\CL_0(H)$.
  \end{enumerate}
\end{thm}

Furthermore, it is shown in \cite{Yang:2023} that $\CL_T(H)$ is non-injective under a spectral density condition on $\Delta_H$ (this result does not require $\dim\Hil<\infty$).

It must be noted that the above results do not hold for $\|T\|=1$. For example, for $T=F$ we have the center $\CL_F(H)\cap\CL_F(H)'=\CL_F(H\cap H')$, which is typically non trivial \cite{LeylandsRobertsTestard:1978}.

\section{Inclusions of twisted Araki-Woods algebras and applications in constructive QFT}\label{sec:QFT}

In this section we sketch how twisted Araki-Woods algebras appear in the construction of integrable quantum field theories on two-dimensional Minkowski spacetime ${\mathbb R}^2$. We will have to confine ourselves to the main ideas, and refer to the review \cite{Lechner:AQFT-book:2015} for more details.

Out of the various axiomatizations of QFT, the operator-algebraic approach \cite{Haag:1996,Araki:1999,Advances-AQFT:2015} is most useful here. In this setting, one models a quantum field theory on $\Rl^2$ by a net of local von Neumann algebras $\CO\mapsto\CA(\CO)$ on a vacuum Hilbert space $\CV$, that is a collection of von Neumann algebras $\CA(\CO)\subset\CB(\CV)$ indexed by (a suitable subset of) all open sets $\CO\subset\Rl^2$.

The minimal physical requirements are that $\CV$ carries a unitary positive energy representation $U$ of the Poincaré group $P$ of $\Rl^2$ and an invariant vector $\Om\in\CV$ (the vacuum vector), such that the following properties hold:
\renewcommand{\labelenumi}{\normalfont\alph{enumi})}
\begin{enumerate}
  \item (Isotony): $\CO_1\subset\CO_2\Rightarrow\CA(\CO_1)\subset\CA(\CO_2)$,
  \item (Locality): $\CA(\CO_1)$ and $\CA(\CO_2)$ commute elementwise if $\CO_1$ lies spacelike to $\CO_2$,
  \item (Covariance): $U(g)\CA(\CO)U(g)^{-1}=\CA(g\CO)$ for every $g\in P$,
  \item (Reeh-Schlieder property): The vacuum vector $\Om$ is cyclic for every $\CO$ with interior points. It is then separating for all $\CO$ such that the causal complement $\CO'$ has interior points.
\end{enumerate}

The task of constructive algebraic QFT is to describe explicit examples of such data based on physical input (see \cite{Summers:2011}). In the case at hand, the aim is to construct a quantum field theory with a presribed elastic two-particle S-matrix.

Such an S-matrix amounts exactly to twists of the form discussed in Examples~\ref{ex:specpar} and \ref{example:crossqft}: The physical meaning of the variable $\te$ is the rapidity of a massive particle, and the fact that $S$ only depends on differences reflects Lorentz symmetry. Our abstract crossing symmetry captures precisely the crossing symmetry of scattering theory in this case, and the modular conjugation $J_L$ of the internal space $L$ models charge conjugation.

As explained in Example~\ref{example:bw}, the unitaries acting by translations in the rapidity form the modular group of a standard subspace that models localization in the Rindler wedge $W\subset\Rl^2$. We may therefore begin by defining the observable algebras of our QFT by setting $\CA_S(W):=\CL_{T_S}(H(W))$, where $T_S$ is the twist based on the two-particle S-matrix $S$, and $H(W)$ the standard subspace given by the wedge $W$.

It then turns out one can easily define the observable algebras for all Poincaré transformed wedges $\Lambda W+x$ by covariance: For translates of $W$, one gets left twisted Araki-Woods algebras $\CL_{T_S}(H(W+x))$, and for the opposite wedges $-W+x$, one arrives at right twisted Araki-Woods algebras $\CR_{T_S}(H(-W+x))$. The observable algebras for bounded regions, such as intersections of two opposite wedges, are then formed by intersecting left and right Araki-Woods algebras, namely the relative commutants $\CA_S(W\cap(-W+x))$ of the inclusions $\CL_{T_S}(H(W+x))\subset\CL_{T_S}(H(W))$, $x\in W$. This construction is perfectly covariant and local, but it is difficult to explicitly exhibit elements of $\CA_S(W\cap(-W+x))$. 

Depending on the details of $S$, it has been shown that the local observable algebra $\CA_S(W\cap(-W+x))$ contains non-trivial operators (functions of the quantum fields defining the model), see \cite{Lechner:AQFT-book:2015} for an overview of results. Once this existence of local observables is settled, one can also prove that the constructed QFT is indeed integrable in the sense that no particle production processed occur in scattering, the $n$-particle S-matrix factorizes into two-particle collisions, and the two-particle S-matrix is given by $S$. Hence this construction solves the inverse scattering problem for $S$.

From the abstract point of view taken for most of this review, what is currently missing is a general understanding for which twists $T$ and for which inclusions $K\subset H$ of standard subspaces the inclusion of von Neumann algebras $\CL_T(K)\subset\CL_T(H)$ has a large relative commutant, for instance in the sense that the Fock vacuum is cyclic for it. The analysis of these inclusions is therefore a subject of ongoing research.

\subsection*{Acknowledgment}

The author thanks the organizers of the XL Workshop on Geometric Methods in Physics 2023 for running a stimulating and enjoyable conference, and Ricardo Correa da Silva for an inspiring and productive collaboration. Financial support by the WGMP and the German Research Foundation DFG through the Heisenberg project ``Quantum Fields and Operator Algebras'' (LE 2222/3-1) is gratefully acknowledged.

% ------------------------------------------------------------------------
\end{document}